\documentclass[10pt,conference,letterpaper]{IEEEtran}

\usepackage{cite}
\usepackage{latexsym}
\usepackage{amsmath}
\usepackage{amsthm}
\usepackage{amssymb}
\usepackage{mathrsfs}
\usepackage[dvips,dvipdf]{graphicx}
\usepackage{graphics}
\usepackage{algorithm}
\usepackage{algorithmic}
\usepackage{multirow}
\usepackage{bbm}
\usepackage{arydshln}
\usepackage{url}
\usepackage{pict2e,picture} 
\usepackage{tikz}   
\usepackage{pifont} 
\usepackage{enumitem}

\usepackage[skip=0pt]{caption}
\usepackage[skip=0pt]{subcaption}

\newtheorem{thm}{Theorem}

\newtheorem{lem}{Lemma}

\newtheorem{assum}{Assumption}

\IEEEoverridecommandlockouts
\begin{document}
	
\title{Toward Low-Cost and Stable Blockchain Networks}

\author{Minghong Fang \mbox{\hspace{0.4cm}} Jia Liu\\ 
	Department of Computer Science\\ 
	Iowa State University
	\thanks{This work has been supported in part by NSF grants ECCS-1818791, CCF-1758736, CNS-1758757; ONR grant N00014-17-1-2417, and AFRL grant FA8750-18-1-0107.}
}

\maketitle

\begin{abstract}
Envisioned to be the future of secured distributed systems, blockchain networks have received increasing attention from both the industry and academia in recent years. 
However, blockchain mining processes demand high hardware costs and consume a vast amount of energy (studies have shown that the amount of energy consumed in Bitcoin mining is almost the same as the electricity used in Ireland). 
To address the high mining cost problem of blockchain networks, in this paper, we propose a blockchain mining resources allocation algorithm to reduce the mining cost in PoW-based (proof-of-work-based) blockchain networks. 
We first propose an analytical queueing model for general blockchain networks. 
In our queueing model, transactions arrive randomly to the queue and are served in a batch manner with unknown service rate probability distribution and agnostic to any priority mechanism. 
Then, we leverage the Lyapunov optimization techniques to propose a dynamic mining resources allocation algorithm (DMRA), which is parameterized by a tuning parameter $K>0$. 
We show that our algorithm achieves an $\left[ {O({1 \mathord{\left/ {\vphantom {1 K}} \right. \kern-\nulldelimiterspace} K}),O(K)} \right]$ cost-optimality-gap-vs-delay tradeoff. 
Our simulation results also demonstrate the effectiveness of DMRA in reducing mining costs.
\end{abstract}

\section{Introduction} \label{sec:intro}

In recent years, blockchain technologies have gained significant attention from both the industry and academia.
In stark contrast to traditional Peer-to-Peer (P2P) networks, blockchain technologies provide a truly decentralized and reliable mechanism, which allows users to communicate securely with each other without the need to trust any third parties. 
As a result, blockchains have been used as a platform for large-scale distributed systems to enable a wide range of decentralized applications (dApps), such as digital currencies \cite{nakamoto2008bitcoin}, food traceability \cite{kamath2018food}, just to name a few.

However, as blockchain applications become more prevalent and the daily transactions in blockchains continue to rise, several fundamental performance issues of blockchain systems also emerge. 
One pressing performance issue is the high cost resulting from skyrocketing mining energy expenditure in operating blockchain systems.
For example, in PoW-based (proof-of-work-based) blockchains such as Bitcoin, a miner has to invest significant computational power to find a solution and show PoW faster than other competing miners.
Due to this competitive nature and the fact that the probability of mining the next block is positively correlated to the computational resources used for solving the puzzle, miners tend to consume vast amounts of electricity and buy various types of hardware (e.g., CPU, GPU, FPGA, ASIC, etc.) to accelerate their mining processes. 
Indeed, studies have shown that the amount of global energy consumed in Bitcoin mining is almost the same as the electricity used in Ireland\cite{o2014bitcoin}.

To achieve a more sustainable growth of blockchain systems, several new mechanisms have been proposed to increase the mining resource efficiency, including, but not limited to, proof-of-stake (PoS)\cite{bentov2016snow}, Trusted Execution Environments (TEEs)\cite{zhang2017rem}, etc.
However, recent studies also found that these new mechanisms either are vulnerable to adversarial attacks or only have marginal energy savings\cite{ahmed2018don}.
Due to the limitations of these new blockchain designs, in this paper, we dedicate ourselves to improving the mining resource efficiency of the traditional and dominant PoW-based blockchains\footnote{PoW-based blockchain protocols account for more than 90\% of existing digital currencies market \cite{gervais2016security}.
}.
Specifically, we focus on optimizing the operations of PoW-based blockchains from a theoretical perspective by taking a {\em queueing theoretic approach} and investigate how to dynamically allocate mining resources based on the number of transactions in the Mempool to minimize the time-average mining cost (mining energy cost offset by reward), while keeping the blockchain networks {\em stable} at the same time.
The main technical results are summarized as follows:

\begin{list}{\labelitemi}{\leftmargin=1em \itemindent=-0.5em \itemsep=.2em}
\item We develop a new logical queueing-based analytical model that allows us to formulate the blockchain mining cost minimization problem as a stochastic optimization problem.
We note that our proposed queueing-based analytical model could be of independent interest for other research problems in blockchain systems.

\item Based on the above analytical model, we propose a queue-length-based online distributed dynamic mining resources allocation algorithm (DMRA) to identify near-optimal solutions.
Our proposed algorithm does not require any prior statistical information of the arrival and service processes. 
Through theoretical analysis, we show that our algorithm comes within an $O(1/K)$ distance to the optimal time-average mining cost, where $K>0$ is a system parameter in our proposed algorithm.
Meanwhile, the proposed algorithm achieves an $O(K)$ bounded time-average queue-length.

\item Besides theoretical analysis, our simulation results also show that DMRA significantly reduces the mining  cost while keeping the blockchain queue stable at the same time.
\end{list}

\section{Related Work} \label{sec:related_work}

To our knowledge, in the blockchain literature, queueing related works include\cite{kasahara2016effect,ricci2019learning,koops2018predicting}.
In \cite{kasahara2016effect}, the authors studied the problem of transaction confirmation time, where the transaction processes are analytically modeled as an $M/{G^B}/1$ queue. 
In this queueing model, transactions arrive at a single server queue according to a Poisson process and are served in batch. 
In their model, miners are profit-driven and select transactions from Mempools according to the transaction fees. 
It was shown that the confirmation time of transactions with a larger transaction fee are shorter than those with a smaller fee. 
Ricci et al. \cite{ricci2019learning} adopted an $M/G/1$ queueing model to analyze transaction delays in blockchain networks. 
They showed that transaction fees and transaction values are two important factors that influence the transaction delays. 
Koops \cite{koops2018predicting} modeled the Bitcoin confirmation time as a Cramer-Lundberg process and provided a dynamic approach to uncover the probabilistic distribution of the confirmation time of Bitcoin transactions.  
Our work differs from \cite{kasahara2016effect,ricci2019learning,koops2018predicting} in the following key aspects:
First, all the aforementioned existing work studied transactions confirmation time, while the focus of this paper is on mining cost minimization subject to the queueing stability of the blockchain system.
Second, unlike \cite{kasahara2016effect,ricci2019learning,koops2018predicting}, 
our proposed algorithm does not rely on any statistical assumptions on the arrivals and services.
This is desirable because, in practice, transactions arrival and service processes are often unknown and can hardly be assumed as a particular queueing model. 

Recently, there is also a line of work focusing on resource allocation problems in blockchain networks. 
For example, Jiao et al. \cite{jiao2018social} considered a mobile blockchain network, where mobile users buy resources from the edge computing service provider (ESP) and the search for PoW solution is offloaded to the service provider. 
They provided an auction-based scheme to maximize the social welfare of mobile blockchain network. 
Similarly, Xiong et al. \cite{xiong2018optimal} also adopted the computational offloading mechanism to manage the computing resource in mobile blockchain. 
They formulated a two-stage Stackelberg game between the service provider and mobile miners to maximize the profit of miners and edge service provider.
Li et. al \cite{li2019decentralized} also modeled the energy allocation as a Stackelberg game for mobile blockchain in Internet-of-Things (IoT) devices, and they applied backward induction to guarantee the benefit of microgrids and miners.
We note that all of these existing works\cite{jiao2018social,xiong2018optimal,li2019decentralized} aimed to study the resources allocation problem of mobile blockchain network in IoT devices, where the mining process is offloaded to the service provider. 
In contrast, we propose a {\em logical} queueing-based analytical model to investigate how to directly allocate mining resources among miners to reduce the their costs in a general blockchain network.

\section{Analytical Model and Problem Formulation} \label{sec:method}

\subsection{Blockchain Network Information Propagation: A Primer} \label{subsec:propagation}

\begin{figure}[t!]
	\centering
	\begin{subfigure}[b]{0.22\textwidth}
		\includegraphics[width=\textwidth]{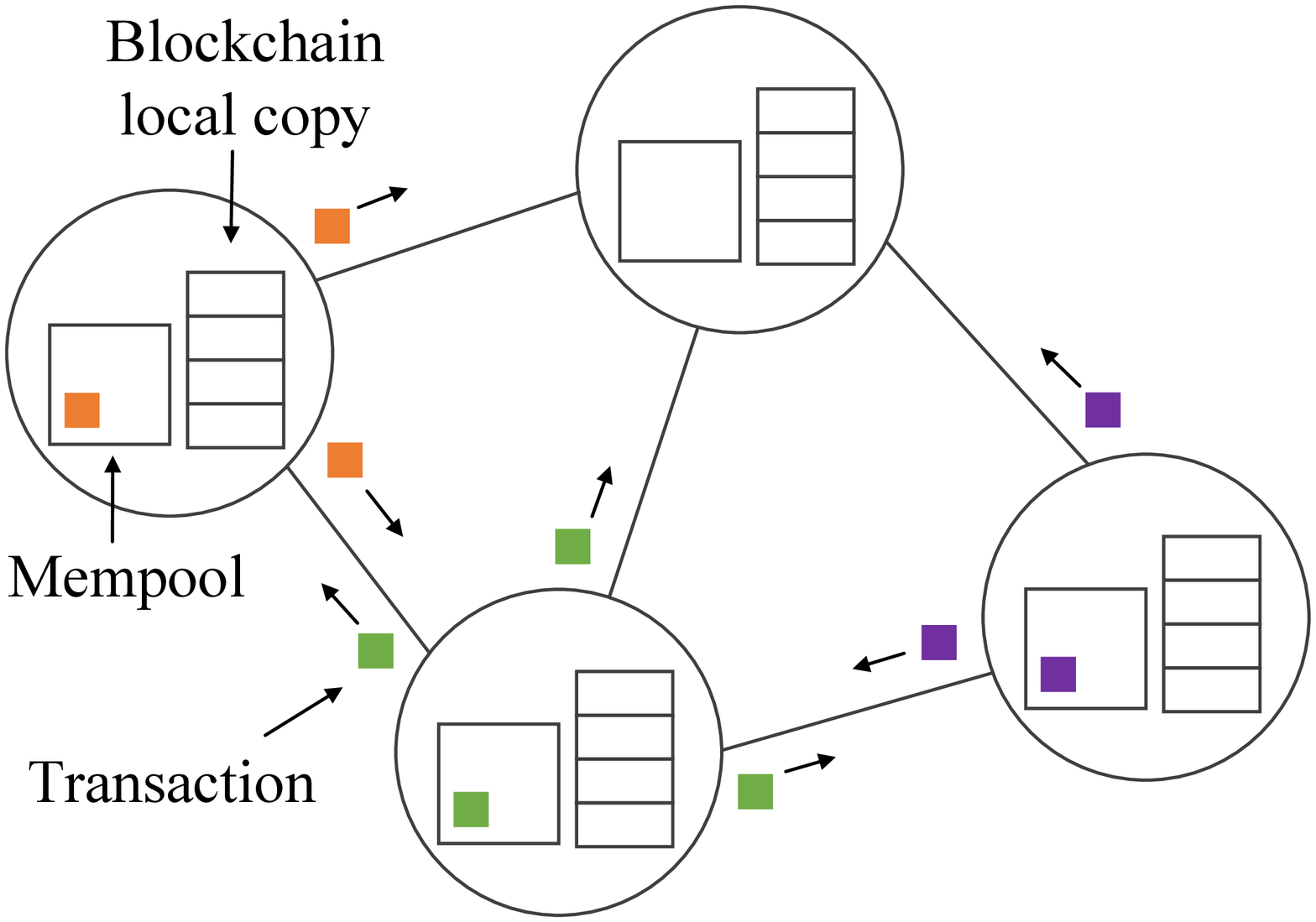}
		\caption{}
	\end{subfigure}
	\begin{subfigure}[b]{0.22\textwidth}
		\includegraphics[width=\textwidth]{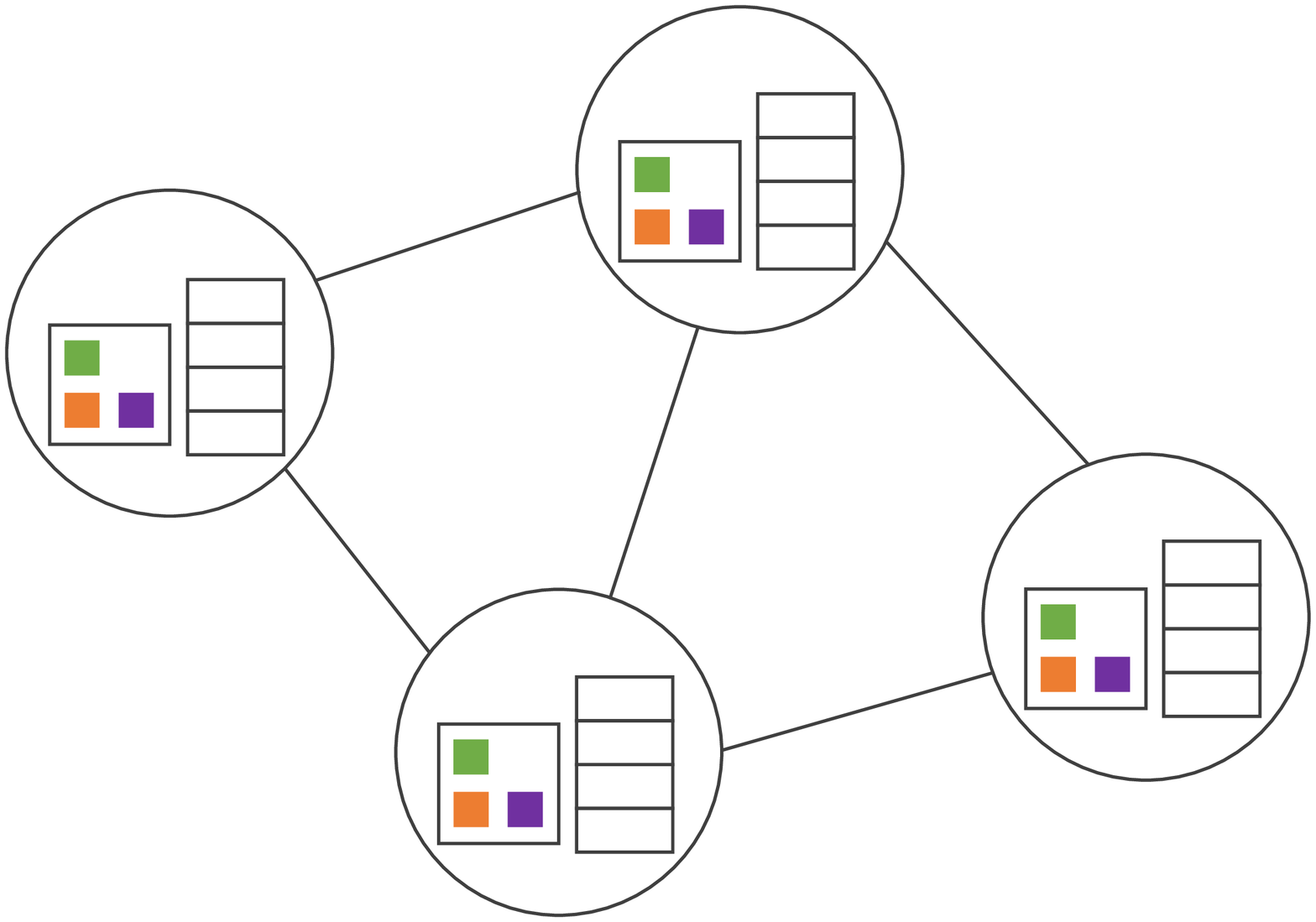}
		\caption{}
	\end{subfigure}
	\begin{subfigure}[b]{0.22\textwidth}
		\includegraphics[width=\textwidth]{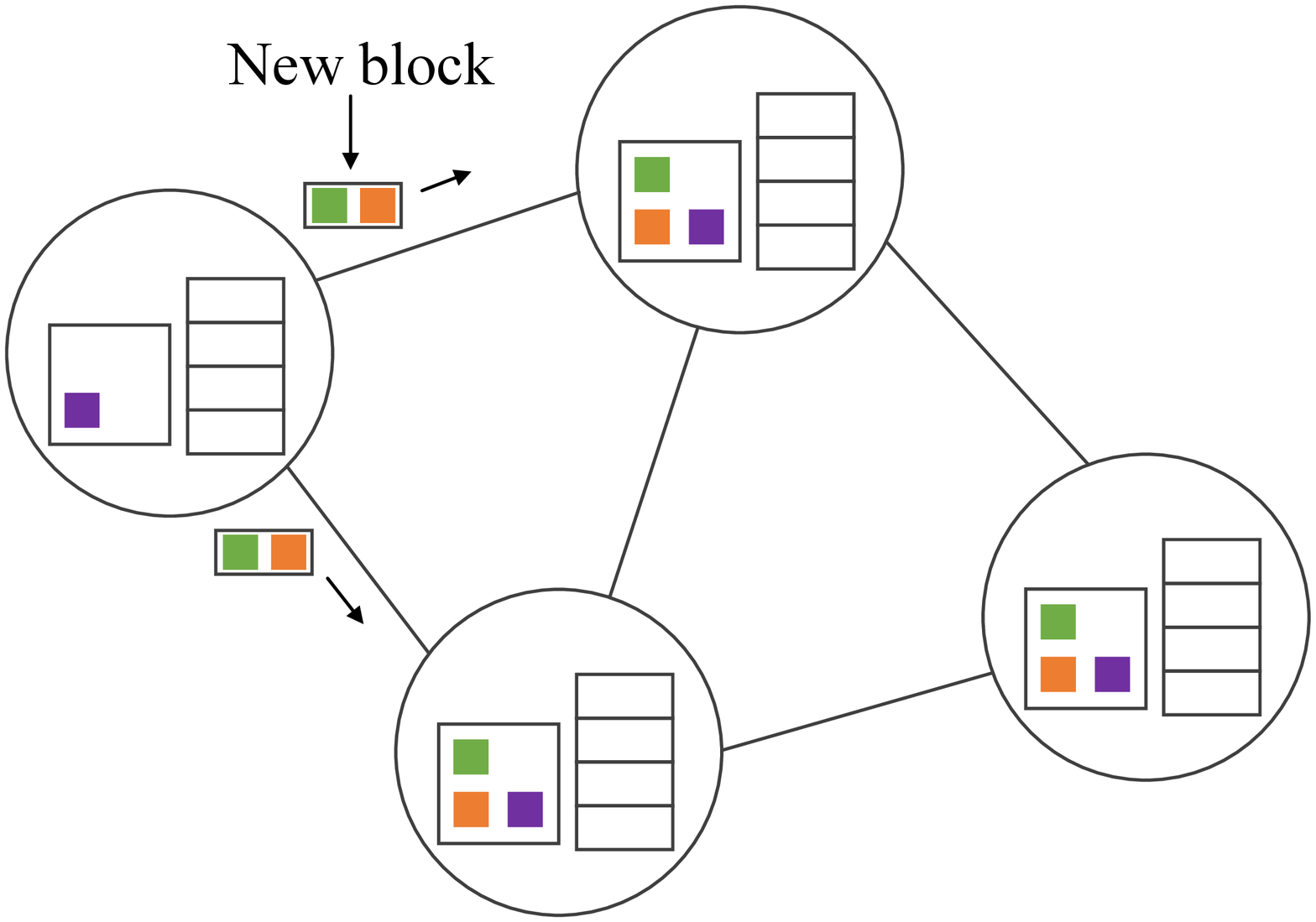}
		\caption{}
	\end{subfigure}
	\begin{subfigure}[b]{0.22\textwidth}
		\includegraphics[width=\textwidth]{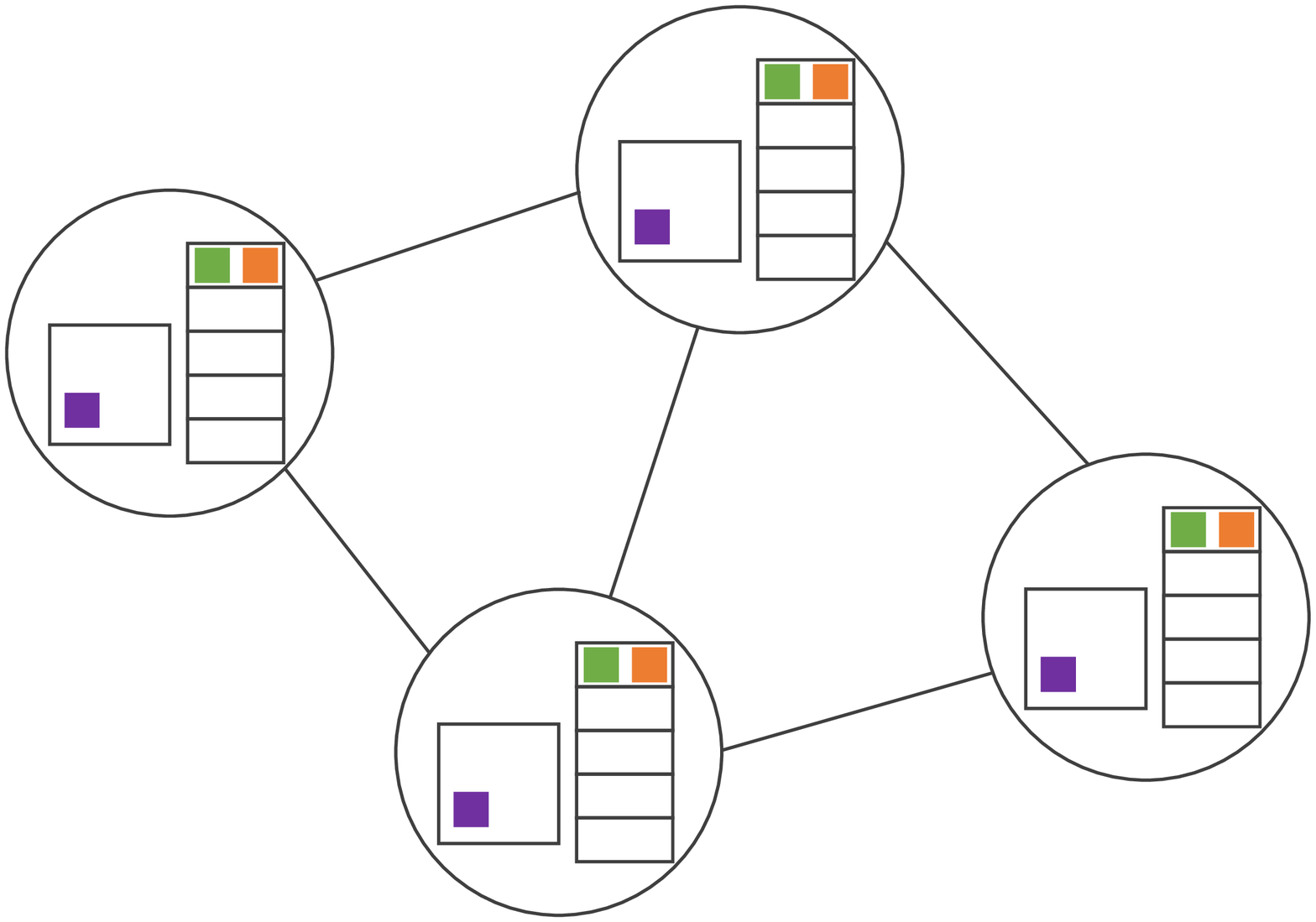}
		\caption{}
	\end{subfigure}
	\caption{An illustrative example of information propagation in a blockchain network.} \label{info_prop}
\vspace{-.2in}
\end{figure}

In blockchain networks, there are four steps for transactions to be delivered to the network and finally included in the global blockchain, as shown in Fig.~\ref{info_prop}. 
Each miner in the network (also referred to as a ``node'' in the rest of the paper) contains an exact local copy of the blockchain, and each miner holds a memory, called Mempool, where transactions are stored and waiting to be processed.
When a new transaction is broadcast to the network, it will first be validated by the miner at which it arrives. 
If valid, the transaction will be stored in the Mempool and propagated to the other miners, see Fig.~\ref{info_prop}(a); otherwise, the entry miner will reject the transaction. 
Next, this transaction is broadcasted to the whole blockchain networks, see Fig.~\ref{info_prop}(b). 
After selecting a set of transactions from their own Mempools and forming them into a block, miners compete to find a solution to PoW. 
When a miner successfully obtains a solution, it immediately broadcasts the newly mined block to its peers, see Fig.~\ref{info_prop}(c). 
Other miners then check whether the solution to PoW is correct. 
If the majority of miners reach a consensus that the solution is correct, the new block will be added to the global blockchain and transactions contained in this verified block will be removed from the Mempools of all miners, see Fig.~\ref{info_prop}(d).

\subsection{A Virtual Blockchain Queueing Model} \label{subsec:queue_model}

Note that in the information propagation process as shown in Fig.~\ref{info_prop},
the propagation time for spreading new transactions across the network in a blockchain is on a much shorter timescale (e.g., seconds) compared to the block mining time (e.g., on the order of 10 minutes). 
This observation suggests that the propagation delays in a blockchain are negligible in practice.
Hence, a new valid transaction can be viewed as arriving at the Mempools of all nodes simultaneously.
Also, since all nodes maintain the same set of transactions after each propagation process is finished, the Mempools of all nodes can be viewed as {\em a global logical queue} that stores all transactions waiting to be processed at multiple servers that represent the miners, as illustrated in Fig.~\ref{block_queue}(a).
In this multi-server queueing model, transactions arrive randomly and are removed from the queue when a new block is generated.
Again, noting that the update time upon recovering a new block across the network is negligible, the queueing model in Fig.~\ref{block_queue}(a) can be further simplified as a single-server queue as shown in Fig.~\ref{block_queue}(b), where the single virtual server models the fact that in each time-slot, only one miner will be winning in solving the PoW.

We assume that the system operates in slotted time indexed by $t \in \{ 0,1,2,...\}$.
We assume that the duration of a time-slot is longer than the average block generation time.
Let $\mathcal{N}$ be the set of miners with $|\mathcal{N}| = N$, and each miner is indexed by $i=1,\ldots,N$.
We let $D$ be the total number of resource types in the system, e.g., electricity, CPU cores, GPU units, FPGA units, etc.
Let $\theta _i^k[t]$, $k \in \{1,2,...,D\}$, denote the amount of resource $k$ allocated at miner $i$ in time-slot $t$.
For convenience, we use vector ${\pmb{\theta} _i}[t] \triangleq [\theta _i^1[t], \theta _i^2[t], ..., \theta _i^D[t]]^{\top} \!\in\! \mathbb{R}^{D}$ to compactly represent all resources that miner $i$ invests to mine a new block in time-slot $t$.
We let weight vector $\pmb{w} \triangleq [w_1,w_2,...,w_D]^{\top} \in \mathbb{R}^{D}$ represent the relative importance of each type of resource. 
Let $Z({\pmb{\theta} _i}[t])$ denote the {\em random} amount time for miner $i$ to mine a new block given resources ${\pmb{\theta} _i}[t]$.
According to \cite{dimitri2017bitcoin}, we assume that $Z({\pmb{\theta} _i}[t])$ follows an exponential distribution with rate parameter ${\lambda _i}({\pmb{\theta} _i}[t])$ dependent on resource allocation $\pmb{\theta}_i[t]$.
We assume ${\lambda _i}({\pmb{\theta} _i}[t])$ is a concave function of ${\pmb{\theta} _i}[t]$, which represents the ``diminishing return effect." 
This is because blockchain systems are fundamentally security-constrained: certain level of latency is required in blockchain networks to mitigate double-spending frauds. 
As a result, doubling the mining resources does not necessarily result in a mining speed twice as fast. 
In this paper, we propose a novel $\epsilon$-parameterized function to model ${\lambda _i}({\pmb{\theta} _i}[t])$:
\begin{align}
{\lambda _i}(\pmb{\theta}_{i}[t]) = {({\pmb{w}^\top {\pmb{\theta}_{i}[t]}})^\epsilon },
\end{align}
where $0 <\epsilon < 1$ is a system parameter.
Clearly, $\lambda_i(\cdot)$ is a concave function. 
Here, $\epsilon$ can be interpreted as a ``difficulty'' parameter of blockchain networks. 
Under security constraints, the larger the value of $\epsilon$, the easier to mine a new block.

\begin{figure}[t!]
	\centering
	\includegraphics[width=.9\linewidth]{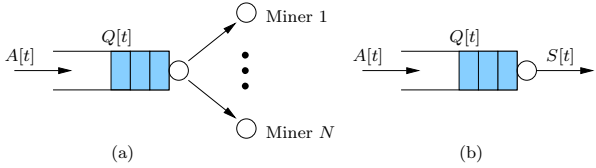}
	\caption{An illustration of blockchain virtual queueing model.}\label{block_queue}
	\vspace{-.2in}
\end{figure}

In blockchain networks, every miner selects a set of transactions from their own Mempools (or the global logical queue in our model) to generate a new block. 
Every miner will work on a different puzzle that is unique to the block they form.
Let $S[t]$ be the random number of blocks mined by the blockchain system in time-slot $t$.
Note that the average time the miners generate a block can be computed as ${\mathbb{E}\{ \mathop {\min }_{i \in \mathcal{N} } Z({\pmb{\theta} _i}[t])\} }$.
Hence, the average number of mined blocks under resources $\pmb{\theta}_{i}[t]$, $\forall i$, can be computed as:
\begin{align}
\mathbb{E}\{ S[t]\}  
& = \frac{1}{{\mathbb{E}\{ \mathop {\min }\nolimits_{i \in \mathcal{N} } Z({\pmb{\theta} _i}[t])\} }} \nonumber\\
&\stackrel{(a)}= \frac{1}{{\frac{1}{{\sum_{i = 1}^N {{\lambda _i}({\pmb{\theta} _i}[t])} }}}} = \sum_{i = 1}^N {{\lambda _i}({\pmb{\theta} _i}[t])},
\label{service_rate}
\end{align}
where $(a)$ follows from the fact that if ${X_1},{X_2},...,{X_N}$ are mutually
independent random variables and ${X_i} \sim \exp ({\lambda _i})$, $\forall i=1,\ldots,N$, then $\min \{ {X_1},{X_2},...,{X_N}\}  \sim \exp (\sum_{i = 1}^N {{\lambda _i}} )$.
Suppose that each mined block contains a fixed number of transactions (denoted by $V$). 
Then, the system can serve $S[t]V$ transactions in time-slot $t$. 
Meanwhile, we assume that, in each time-slot, new transactions arrive randomly at the blockchain system according to a process $\{A[t]\}_{t=0}^{\infty}$, where $A[t]$ denotes the number of transactions arriving in time-slot $t$. 
In this paper, we assume that $A[t] \leq A_{\max}$ and $S[t] \leq S_{\max}$, $\forall t$, for some constants $A_{\max}, S_{\max} \geq 0$.
Let $Q[t]$ represent queue-length (i.e., the number of transactions) in the queue in time-slot $t$. 
Then, the queue-length evolution can be written as:
\begin{align}
Q[t{\rm{ + }}1] = \max \left\{ {Q[t] - S[t] V + A[t],0} \right\}.
\label{actual_queue}
\end{align}

\subsection{Problem Formulation} \label{subsec:formulation}

In PoW-based blockchains, miners compete to solve the puzzle to mine a new block, and the successful miner will win a fix reward $R$ and flexible transaction fees $M$. 
Since there are $S[t]$ number of blocks mined by the blockchain system in time-slot $t$, the total reward provided by the system is $S[t](R+M)$.
We let $c(\pmb{\theta}_{i}[t])$ denote the cost of miner $i$ in time-slot $t$ given resource allocation $\pmb{\theta}_{i}[t]$, where the cost function $c(\cdot)$ is assumed to be increasing, convex, and differentiable. 
We note that here we assume that the miner will receive the reward once he mines a new block.
In practice, the miner may not receive the reward if the mined block is not finally added to the global blockchain due to branching\footnote{
By design, a blockchain could branch if different blocks are generated simultaneously.
Eventually, blocks not on the main chain may be discarded. 
}.
Hence, our model can be viewed as a lower bound for the mining cost in actual blockchain networks. 
The problem of handling cost in branching will be left for our future studies.

From a social welfare perspective, the blockchain system operator wants the miners to invest fewer total resources in mining new blocks.
At the same time, each miner tries to maximize his/her own rewards.
These two conflicting goals can be represented by the total net expected mining cost in time-slot $t$ (resource costs subtracted by the total received rewards for the miners $R+M$), which can be computed as:
\begin{align}
\mathbb{E}\{p[t]\} & = \mathbb{E} \bigg\{ \sum\nolimits_{i = 1}^N c(\pmb{\theta}_{i}[t]) - S[t](R+M) \bigg\}. 
\notag \\
& = \sum\nolimits_{i = 1}^N c(\pmb{\theta}_{i}[t]) - (R+M) \sum\nolimits_{i = 1}^N {{\lambda _i}({\pmb{\theta} _i}[t])}. 
\label{total_cost}
\end{align}
Define $\overline{p}$ as the time-average mining cost for all miners, i.e., $\overline{p} = \limsup_{t \to \infty } \frac{1}{t} \sum_{\tau = 0}^{t - 1} \mathbb{E}\left\{ p[\tau] \right\}$.

Apart from minimizing the time-average mining cost, the blockchain system operator also wants to guarantee that the allocated mining resources at all miners are sufficient such that transactions can be processed in a timely manner and the  blockchain queue is stable. 
In this paper, we adopt the following notion of queueing stability\cite{neely2010stochastic}: We say that the blockchain queue is strongly stable if it satisfies:
\begin{align} 
\mathop {\limsup }\limits_{t \to \infty } \frac{1}{t}\sum\limits_{\tau  = 0}^{t - 1} {\mathbb{E}\left\{ Q[\tau ]\right\} }  < \infty, 
\label{queue_stable}
\end{align}
that is, the stability of blockchain queue implies that the transactions in queue cannot be accumulated to infinity. 
Putting the above objective function and constraints together, the mining resources allocation optimization problem (MRA) can be formulated as:
\begin{equation*}
\begin{array}{lll}
\textbf{MRA:} &\text{Minimize } & \!\!\!\! \overline{p} \\
& \text{subject to } & \!\!\!\! 0 \le {\pmb{\theta} _i}[t] \le {\pmb{\theta} _{\max}}, \forall i, t, \\
& & \!\!\!\! \text{Queueing-stability constraint in } (\ref{queue_stable}),
\end{array} \!\!\!\!
\end{equation*}
where vector ${\pmb{\theta} _{\max}}$ contains the upper bounds for all mining resources and the inequalities between vectors in the constraints are component-wise.
In what follows, we will develop a dynamic resource allocation algorithm that provably achieves the optimal solution to Problem MRA.

\section{Dynamic Mining Resources Allocation Algorithm} \label{sec:alg}

In this section, we consider the problem of optimally allocating mining resources to minimize the mining cost while maintaining the stability of the blockchain system.
To this end, we propose a dynamic mining resources allocation algorithm in Section~\ref{subsec:dmra} and present theoretical analysis of the proposed algorithm in Section~\ref{subsec:analysis}.

\subsection{Dynamic Mining Resources Allocation Algorithm} \label{subsec:dmra}

We note that Problem MRA is a challenging Markov decision process problem under a dynamic system with a large solution space.
To solve Problem MRA, we first present a useful lemma that shows that under mild assumptions, it suffices to only consider the class of stationary randomized policies, which significantly reduces the size of solution space to be considered. 

\begin{lem}[Solution Space] \label{lem_solnspace}
	Let $\pmb{\theta}^{*} = [\pmb{\theta}_{1}^{*},\ldots,\pmb{\theta}_{N}^{*}]^{\top}$ be an optimal mining resources allocation vector.
	If Problem MRA is feasible, and the blockchain system satisfies the boundedness assumptions and the law of large number assumption,  
	then, for any $\delta \ge 0$, there exists a stationary randomized policy that makes mining resources allocation decisions depending only on the system state while at the same time satisfies: $\bar{p} \leq p^* + \delta$ and $\mathbb{E}\{ A[t]\}  \le \mathbb{E}\{ S[t]V\}$,
where ${p^*} \triangleq \sum_{i = 1}^N {\{ c(\pmb{\theta}_{i}^* ) - (R+M){\lambda _i}(\pmb{\theta}_{i}^*)\}}$.
\end{lem} 

The proof of Lemma~\ref{lem_solnspace} follows from \cite[Theorem 4.5]{neely2010stochastic} and we omit the details here due to space limitation.
Based on this insight, in what follows, we propose a dynamic mining resources allocation algorithm based on Lyapunov optimization technique to stabilize the blockchain queue and solve Problem MRA.
To this end, we first define the following Lyapunov function $L[t] \triangleq \frac{{Q{{[t]}^2}}}{2}$.
Hence, the Lyapunov drift can be computed as $\Delta L[t] = \mathbb{E}\left\{ {L[t + 1] - L[t]|Q[t]} \right\}$.
It is known from the classical work by Tassiulas and Ephremides \cite{Tassiulas92:BackPressure} that greedily minimizes $\Delta L[t]$ induces queueing-stability of the system.
Further, to minimize the mining cost while maintaining queueing stability, 
we incorporate the mining cost function of Problem MRA into the Lyapunov drift $\Delta L[t]$.
That is, in every time-slot $t$, we greedily minimize the following drift-plus-penalty expression:
\begin{align}
\Delta  L[t] + K \sum\nolimits_{i = 1}^N {\left\{ {c(\pmb{\theta}_{i}[t]) - (R+M){\lambda _i}(\pmb{\theta}_{i}[t])|Q[t]} \right\}},
\label{drift_penalty}
\end{align}
where $K >0$ is a parameter that, as will be shown later, is used to tune the tradeoff between queueing delay and the achievable mining cost.
Next, we will state a basic property of the expression in (\ref{drift_penalty}) that will be useful in our later performance analysis:
\begin{lem}[Drift Bound] \label{lem_drift_bnd}
For $0 \le {\pmb{\theta} _i}[t] \le {\pmb{\theta}_{\max}}$, the drift-plus-penalty term in (\ref{drift_penalty}) satisfies:
	\begin{align}
	&\Delta  L[t] + K \sum\nolimits_{i = 1}^N {\left\{ {c(\pmb{\theta}_{i}[t]) - (R+M){\lambda _i}(\pmb{\theta}_{i}[t])|Q[t]} \right\}}   \notag \\
	& \quad\quad \leq B +  K \sum\nolimits_{i = 1}^N {\left\{ {c(\pmb{\theta}_{i}[t]) - (R+M){\lambda _i}(\pmb{\theta}_{i}[t])|Q[t]} \right\}}  \notag \\
	& \quad\quad\quad + \mathbb{E}\left\{ {Q[t](A[t] - S[t]V)|Q[t]} \right\}, 
	\label{ly_opt_min}
\end{align}
	where $B$ is defined as $B \triangleq \frac{{\max \{A_{\max }^2,S_{\max }^2 V^2} \} }{2}$.
\end{lem}
\begin{proof}
	Using the queue-length evolution in (\ref{actual_queue}) and the definition of the conditional Lyapunov drift given $Q[t]$, we have:
	\begin{align}
	\Delta L[t] 	
	&\!= \!\mathbb{E} \! \left\{ {L[t + 1] \!-\! L[t]} | Q[t] \right\} 
	\!=\! \frac{1}{2}\mathbb{E}\left\{ {Q{{[t + 1]}^2} \!-\! Q{{[t]}^2}} | Q[t] \right\} \notag \\
	& \!\!\!\!\!\!\!\!\!\!\!\!\!\!\!\!\!\! \le \!\frac{1}{2}\mathbb{E}\!\left\{ {{{\!(Q[t] \!-\! S[t]V \!+\! A[t])}^2} \!\!\!-\! Q{{[t]}^2}} | Q[t] \right\} 
	\!\!=\! \frac{1}{2}\mathbb{E}\!\left\{ {{{\!(A[t] \!\!-\!\! S[t]V)}^2}\!} \right\} \notag \\
	& \!\!\!\!\!\!\!\!\!\!\!\!\!\!\!\!\!+ \!
	 \mathbb{E}\!\left\{ {Q[t](A[t] \!-\! S[t]V)} | Q[t] \right\} 
	 \!\stackrel{(a)}{\le} B \!+\! Q[t] \mathbb{E}\!\left\{ {A[t] \!-\! S[t]V} |Q[t] \right\}\!,  \nonumber
	\end{align}	
	where $(a)$ follows from $(x - y) \le \max [ x^2, y^2 ]$ for $x \ge 0, y \ge 0$. 
	Adding $ K \sum_{i = 1}^N {\{ {c(\pmb{\theta}_{i}[t]) - (R+M){\lambda _i}(\pmb{\theta}_{i}[t])|Q[t]} \}}$ on both sides completes the proof.
\end{proof}

Since the arrival process $A[t]$ of new transactions is independent of $Q[t]$, minimizing the right-hand-side (RHS) of (\ref{ly_opt_min}) for all $t$, i.e., the upper bound of the drift-plus-penalty expression, is equivalent to solving the following optimization problem by plugging in Eq.~(\ref{service_rate}) and rearranging terms:
\begin{equation} \label{objectiveFrame_dyna_final}
\begin{array}{ll}
\!\!\!\! \text{Minimize } & \!\! \hspace{-.1in} K \!\! \sum\limits_{i = 1}^N \! {c(\pmb{\theta}_{i}[t]) \!-\! [K(R \!+\! M) \!+\! Q[t]V]} \sum\limits_{i = 1}^N {{\lambda _i}(\pmb{\theta}_{i}[t])} \\
\!\!\!\! \text{subject to } & \!\! \hspace{-.1in} 0 \le {\pmb{\theta} _i}[t] \le {\pmb{\theta}_{\max}}. 
\end{array} \!\!\!\!\!\!
\end{equation}
This observation motivates a dynamic mining resource allocation algorithm (DMRA) stated in Algorithm~\ref{resource_alloc_alg} as follows:

\begin{algorithm}[H]
	\caption{Dynamic \!\! Mining \!\! Resource \!\! Allocation \!\! Algorithm.} \label{resource_alloc_alg}
	\begin{algorithmic}[1]
		\STATE Initialization: Let $t=0$; Choose a value for $K$.
		\FOR {each time-slot $t=0,1,\ldots$}
		\STATE Each miner $i$ observes $Q[t]$ and then allocates mining resources $\pmb{\theta}_{i}[t]$, $\forall i$ by solving Problem (\ref{objectiveFrame_dyna_final}).			
		\STATE Let the blockchain queue evolve according to (\ref{actual_queue}).
		\ENDFOR
	\end{algorithmic} 
\end{algorithm}

\subsection{Performance Analysis} \label{subsec:analysis}
In this section, we will analyze our proposed dynamic mining resources allocation algorithm in Algorithm~\ref{resource_alloc_alg}.  
\subsubsection{Upper Bound for Mining Cost}
We first show that our proposed DMRA algorithm achieves a mining cost that is guaranteed to be within an $O(1/K)$ distance to the optimal value of Problem~MRA, and the result is stated in Theorem~\ref{theorem_upperBound_cost}:
\begin{thm} \label{theorem_upperBound_cost} 
Let $B$ be defined as in Lemma~\ref{lem_drift_bnd}.
Suppose that Problem~MRA is feasible such that there exists an optimal solution to Problem~MRA that achieves the optimal value $p^*$.
	For any $K>0$, the time-average mining cost incurred by the DMRA algorithm is upper bounded by:
	\begin{align} 
	\limsup\limits_{t \to \infty } \frac{1}{t}\sum\limits_{\tau  = 0}^{t - 1}\sum\limits_{i = 1}^N {\{ c(\pmb{\theta}_{i}[\tau]) - (R+M){\lambda _i}(\pmb{\theta}_{i}[\tau])\}}
	 \leq p^*  + \frac{B}{{K}}. \notag
	\end{align}		
\end{thm}

\begin{proof}
	To simplify the notation, in the rest of the paper, we use $F[t] \triangleq \sum\nolimits_{i = 1}^N {\left\{ {c(\pmb{\theta}_{i}[t]) - (R+M){\lambda _i}(\pmb{\theta}_{i}[t])} \right\}} $.
	Plugging the result of Lemma~\ref{lem_solnspace} into the RHS of (\ref{ly_opt_min}) under given queue-length $Q[t]$ and setting $\delta \to 0$, we have:
	\begin{align} \label{eqn_thm_step1}
	 \Delta  L[t] \!+\! KF[t]  
	 \leq B \!+\! Kp^* \!+\! Q[t] \mathbb{E}\left\{ A[t] \!-\! S[t]V | Q[t] \right\}.
	\end{align}
	Taking expectation with respect to $Q[t]$ and then summing over $\tau  \in \{ 0,1,...,t-1\}$, we have:
	\begin{align}
	& \mathbb{E} \{L[t]\} - \mathbb{E} \{L[0]\} +  K\sum\nolimits_{\tau  = 0}^{t - 1} F[\tau]  \notag \\
	&  \leq Bt + K{p^*}t  +  \sum\nolimits_{\tau  = 0}^{t - 1} \mathbb{E}\left\{ {Q[\tau](A[\tau]- S[\tau]V)} \right\}.
	\label{upperBound_all}
	\end{align}
	Since $Q[t] \!\ge\! 0$, $\mathbb{E}\left\{ {A[t]} \right\} \!\le\! \mathbb{E}\left\{ {S[t]V} \right\}$ and $L[t] \!\ge\! 0$, we have:
	\begin{align}
	 \frac{1}{t}\sum\nolimits_{\tau  = 0}^{t - 1} F[\tau]
	 \le p^* + {\frac{B}{{K}}} + \frac{\mathbb{E}\{L[0]\}}{{{K}t}}. 
	\label{upperBound_cost}
	\end{align}
	Finally, taking $\limsup$ as $t \to \infty$ and plugging in the definition of $F[t]$ yields the upper bound.
\end{proof}

\subsubsection{Bounded Average Queue Length}
Next, we consider the delay performance of our proposed DMRA algorithm.
To this end, we first introduce the following assumption.
\begin{assum} \label{assump_slater}
	(Slater Condition). For the expected arrival rate and service rates, there exists a constant value $\Delta > 0$ that satisfies: $\mathbb{E}\left\{ {A[t]} \right\} \le \mathbb{E}\left\{ {S[t]V} \right\}  - \Delta$.  
\end{assum}
\begin{thm}
	If Problem~MRA is feasible and Assumption~\ref{assump_slater} holds, then our proposed dynamic mining resources allocation algorithm stabilizes the blockchain queue, i.e., the queue is strongly stable, and the time-average queue-length satisfies:
	\begin{align}
	\mathop {\lim \sup }\limits_{t \to \infty } \frac{1}{t}\sum\limits_{\tau  = 0}^{t - 1} {\mathbb{E}\left\{ Q[\tau]\right\} }  \le \frac{{B + {K}({p^*} - {p^{\min }})}}{\Delta },  \notag
	\end{align}	
	where $p^{\min }$ is the minimum mining cost.
	\label{theorem_queue_length}
\end{thm}
\begin{proof}
	According to (\ref{upperBound_all}), since $L[t]\ge 0$, we have:
	\begin{align}
	\!\!\!-\mathbb{E} \{L[0]\} \!+\!  K\!\!\sum\limits_{\tau  = 0}^{t - 1} F[\tau] 
	\!\! \le \!  \! Bt \!+\! K\!{p^*}t  \!+\! \sum\limits_{\tau  = 0}^{t - 1} \!\mathbb{E}\!\left\{ {Q[\tau](A[\tau] \!-\! S[\tau]V)} \right\} \!\!\notag.
	\end{align}	
	According to Assumption~\ref{assump_slater}, we have:
	\begin{align}
	&\frac{1}{t} \sum\limits_{\tau  = 0}^{t - 1} \mathbb{E}\left\{ {Q[\tau]} \right\} \le \frac{Bt \!+\! \mathbb{E} \{L[0]\} \!+\! {{K{p^*}t \!-\! {K}\sum\nolimits_{\tau  = 0}^{t - 1} F[\tau]}}} {t\Delta } \notag \\
	& \le \frac{{B + {K}({p^*} - {p^{\min }})}}{\Delta } +\frac{{\mathbb{E} \{L[0]\}}}{{t\Delta }} .
	\end{align}
	Taking limsup as $t \to \infty$ completes the proof.
\end{proof}

\subsubsection{Time-Varying $K$ for Asymptotic Optimality}
Rather than using a fixed $K$, we can gradually increase $K$ while maintaining the blockchain queue to be mean state stable. 
By doing so, we can eliminate the term $ \frac{B}{{K}}$ in Theorem \ref{theorem_upperBound_cost}.
\begin{thm} \label{thm_zero_gap}
	If $K[t]$ is used in the drift expression of (\ref{eqn_thm_step1}) and $K[t]$ satisfies $K[t] \triangleq {K_0}{(t + 1)}$, where $K_0 >0$ is a constant, then it holds that
	\begin{align} 
	\limsup\limits_{t \to \infty } \frac{1}{t}\sum\limits_{\tau  = 0}^{t - 1}\sum\limits_{i = 1}^N {\{ c(\pmb{\theta}_{i}[\tau]) - (R+M){\lambda _i}(\pmb{\theta}_{i}[\tau])\}}
	= p^*. \notag
	\end{align}				
	\label{variable_K}
\end{thm}
\begin{proof}
	Since $\Delta L[t] = \mathbb{E}\left\{ {L[t + 1] - L[t]|Q[t]} \right\}$, $\mathbb{E}\left\{ {A[t]} \right\} \le \mathbb{E}\left\{ {S[t]V} \right\}$, according to (\ref{eqn_thm_step1}), replacing $K$ with $K[t]$ and taking expectation with respect to $Q[t]$, then we have:
	\begin{align} 
	\mathbb{E}\{{L[t + 1]}\} - \mathbb{E}\{{L[t]}\} + 
	K[t]F[t]  \leq B + K[t]p^*.
	\end{align} 
	Dividing both sides by $K[t]$, we have:
	\begin{align} 
	\frac{{\mathbb{E}\left\{ {L[t + 1]} \right\}}}{{K[t]}} - \frac{{\mathbb{E}\left\{ {L[t]} \right\}}}{{K[t]}}  + 
	F[t] \le \frac{B}{{K[t]}} + p^* .
	\end{align} 
	Summing over $\tau  \in \{ 0,1,...,t-1\}$, we have:
	\begin{align}
	\label{theorem3_summing}
	&\frac{{\mathbb{E}\left\{ {L[t]} \right\}}}{{K[t - 1]}} + \sum\limits_{\tau  = 1}^{t - 1} {\mathbb{E}\left\{ {L[\tau]} \right\}} \left[ {\frac{1}{{K[\tau - 1]}} - \frac{1}{{K[\tau]}}} \right] + 
	\sum\limits_{\tau  = 0}^{t - 1}F[\tau] \notag \\
	&\le \sum\limits_{\tau  = 0}^{t - 1} {\frac{B}{{K[\tau]}}} + p^*t. 
	\end{align}
	Since $K[\tau]$ is non-decreasing and $K[\tau] > 0$, $\forall \tau$, then we have ${\frac{1}{{K[\tau - 1]}} - \frac{1}{{K[\tau]}}} \ge 0 $.
	Because $L[t] \ge 0$, dividing $t$ on both sides of (\ref{theorem3_summing}), we have 
	$
	\frac{1}{t} \sum\nolimits_{\tau  = 0}^{t - 1} F[\tau] \!\le\! p^* \!+\! \frac{1}{t}\sum\nolimits_{\tau  = 0}^{t - 1} {\frac{B}{{K[\tau]}}}. $
	Taking limsup as $t \to \infty$, we get 
	$
	\frac{1}{t}\sum\nolimits_{\tau  = 0}^{t - 1} {\frac{B}{{K[\tau]}}}  = \frac{B}{{{K_0}}}\frac{1}{t}\sum\nolimits_{\tau  = 0}^{t - 1} {\frac{1}{{\tau {\rm{ + }}1}}}  \le \frac{B}{{{K_0}}}\frac{{\log t}}{t} \to 0
	\label{theom3_p_seriers}
	$, which concludes the proof.
\end{proof}

Theorems~\ref{theorem_upperBound_cost} and \ref{theorem_queue_length} show that our algorithm admits a $\left[ {O({1 \mathord{\left/ {\vphantom {1 K}} \right. \kern-\nulldelimiterspace} K}),O(K)} \right]$ cost-delay tradeoff. 
It achieves within an ${O({1 \mathord{\left/
			{\vphantom {1 K}} \right.
			\kern-\nulldelimiterspace} K})}$ neighborhood of the optimal time-average mining cost with a delay on the order of ${O(K)}$. Theorem~\ref{thm_zero_gap} shows that we can narrow the gap by choosing the sequence	$\{K[t]\}$.

\section{Numerical Results} \label{sec:numerical}

In this section, we conduct numerical simulations to evaluate the performance of our proposed dynamic mining resource allocation algorithm for PoW-based blockchain networks. 

{\em 1) Simulation Settings}:
In our experiment, we consider a blockchain network with $N$ miners. 
Each of these miners generates transactions and mines new blocks with one CPU core. 
Transactions arrivals $A[t]$ are i.i.d. across time and uniformly distributed over the interval $[50, 200]$. 
In this paper, we consider two types of mining resources.
The first one is CPU utilization, e.g., the CPU utilization ranges from 40\% to 60\%, and the CPU utilization is managed by the cpulimit tool. 
Higher CPU utilization means higher hashing rate and thus higher probability that the miner successfully mines a block.
However, higher CPU utilization also results in a higher mining energy cost. 
The second type of resource is the electricity.
In the simulations, we adopt the following linear cost model:
$
c(\pmb{\theta}_{i}[t]) = m(\pmb{w}^\top {\pmb{\theta}_{i}[t]}) + n,
\label{cost_fun_def}
$
where $m,n \ge 0$ are constants.
When computing the reward of each block, we use the fixed value reward ($R$) and ignore the random transaction fee ($M$) for simplicity. 
We compare our proposed DMRA algorithm with two baselines. 
The first baseline is denoted as MaxMining, where the system tries to make use of all available resources to mine a new block at all times. 
The second baseline is denoted as RandMining, where the system randomly allocates mining resources to mine a block.
Unless otherwise mentioned, we use the following default parameter settings: 
The length of each time-slot is 1 minute, the total time duration of all simulations is 200 minutes.
We set the tradeoff parameter $K = 20$.
The weight for CPU utilization and electricity is set to 3 and 1, respectively.
We let $N = 4$, $R = 3$, $V = 3$, $m = 0.45$, $n = 0$, and $\epsilon = 0.5$.

{\em 2) Performance Analysis}:
Figs.~\ref{fig_mine_cost_base} and~\ref{fig_mine_queue_base} show the mining cost and delay performance achieved by our proposed DMRA method and baselines. 
We observe that our dynamic mining resources allocation algorithm DMRA achieves similar delay performance as MaxMining.
However, DMRA significantly reduces the mining cost compared to MaxMining. 
Even though RandMining could reduce the mining cost to some extent, it could not keep the blockchain queue stable since it allocates the mining resources randomly. 
In general, our proposed DMRA method significantly reduces the mining cost while keeping the blockchain queue stable at the same time.
This is because, compared to the baselines, our proposed algorithm can dynamically allocate mining resources according to the number of transactions in the blockchain queue.

\begin{figure}[t!]
	\centering
	\begin{minipage}[t]{0.47\linewidth}
		{\includegraphics[height=3.6cm]{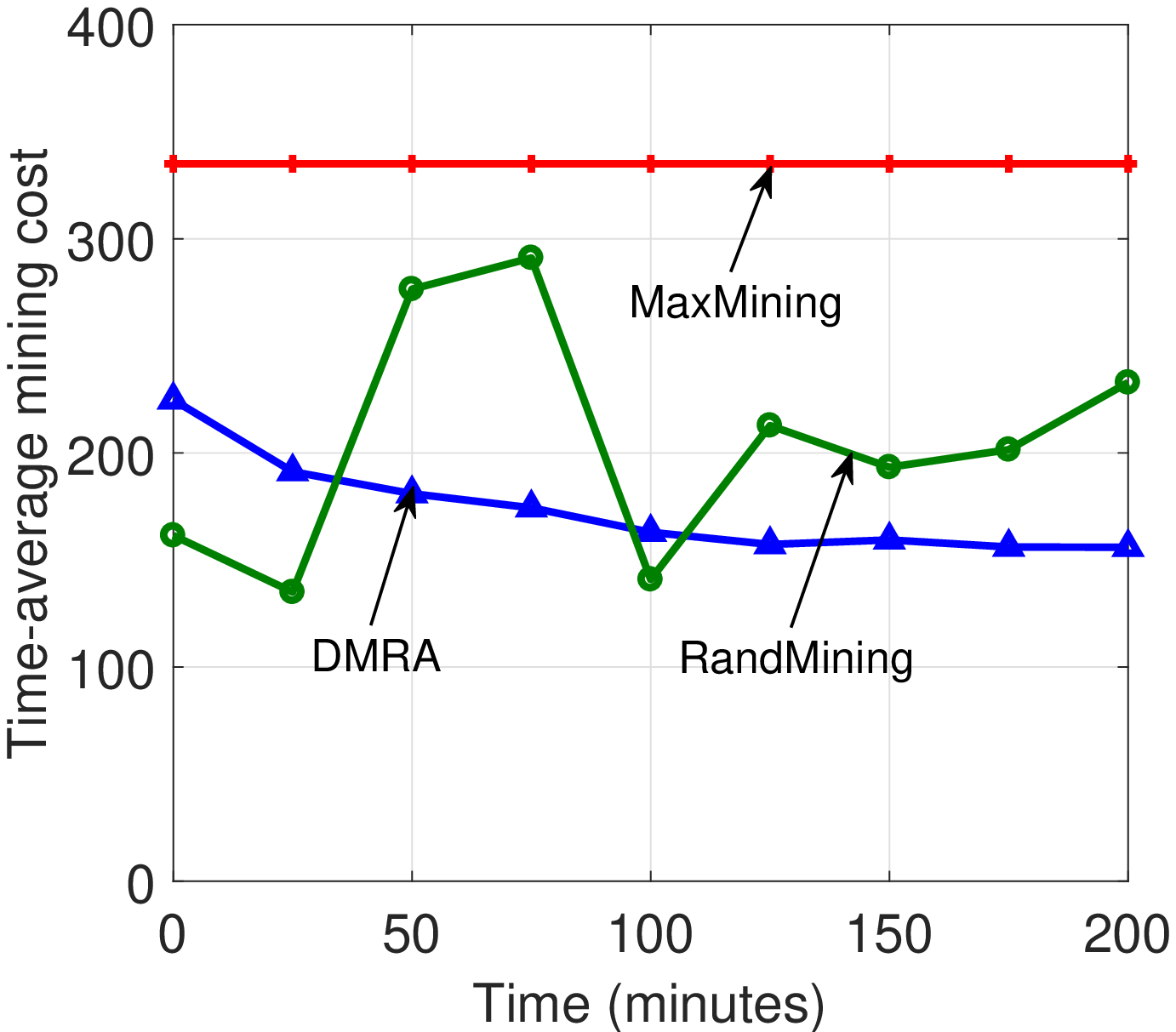}}
		\caption{Time-average mining cost at time slots.} 
		\label{fig_mine_cost_base}
	\end{minipage}%
	\hspace{0.009\textwidth}
	\begin{minipage}[t]{0.47\linewidth}
		{\includegraphics[height=3.6cm]{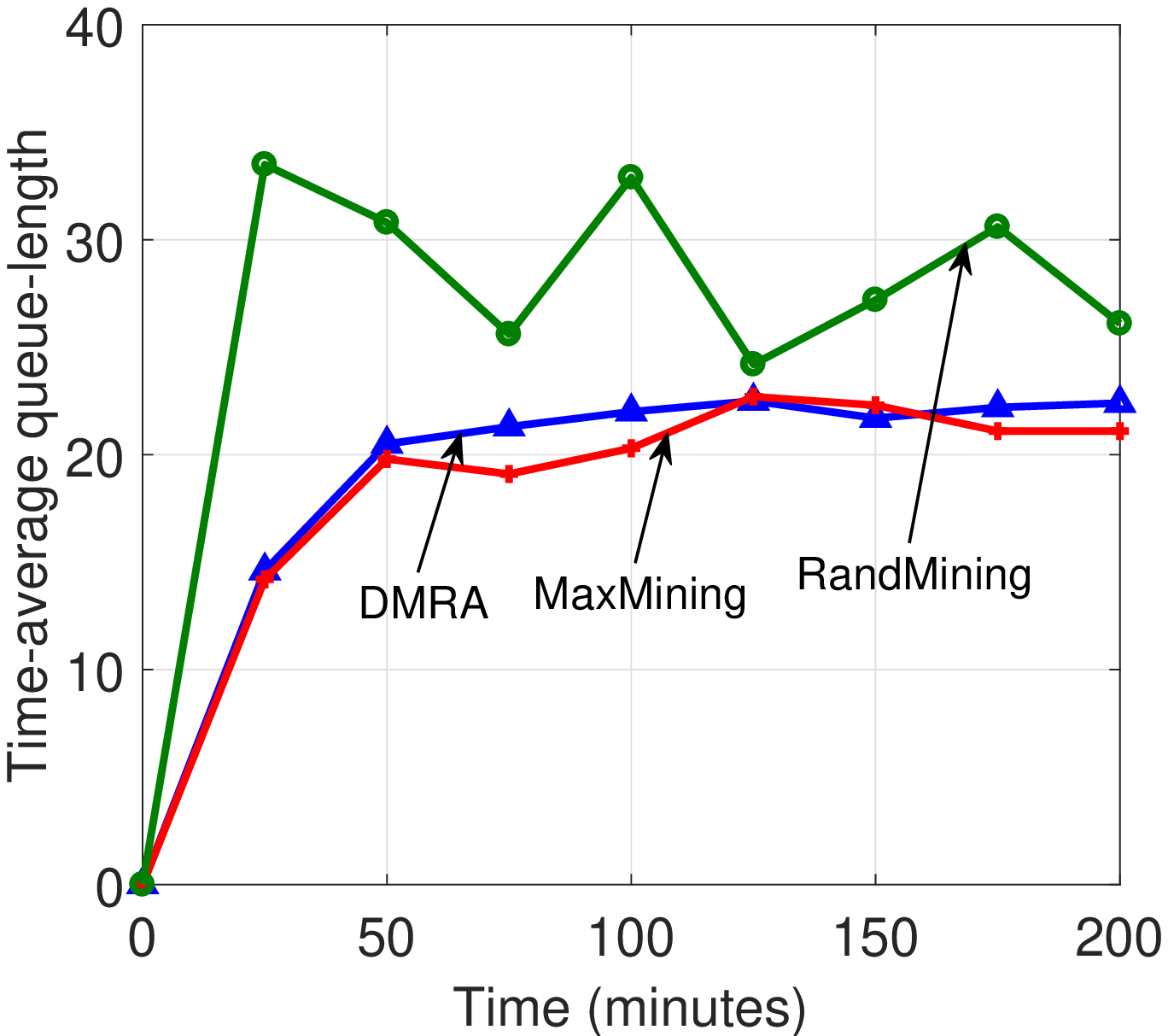}}
		\caption{Time-average queue-length at time slots.} 
		\label{fig_mine_queue_base}
	\end{minipage}
	\vspace{-.15in}
\end{figure}

\begin{figure}[t!]
	\centering
	\begin{minipage}[t]{0.47\linewidth}
		{\includegraphics[height=3.6cm]{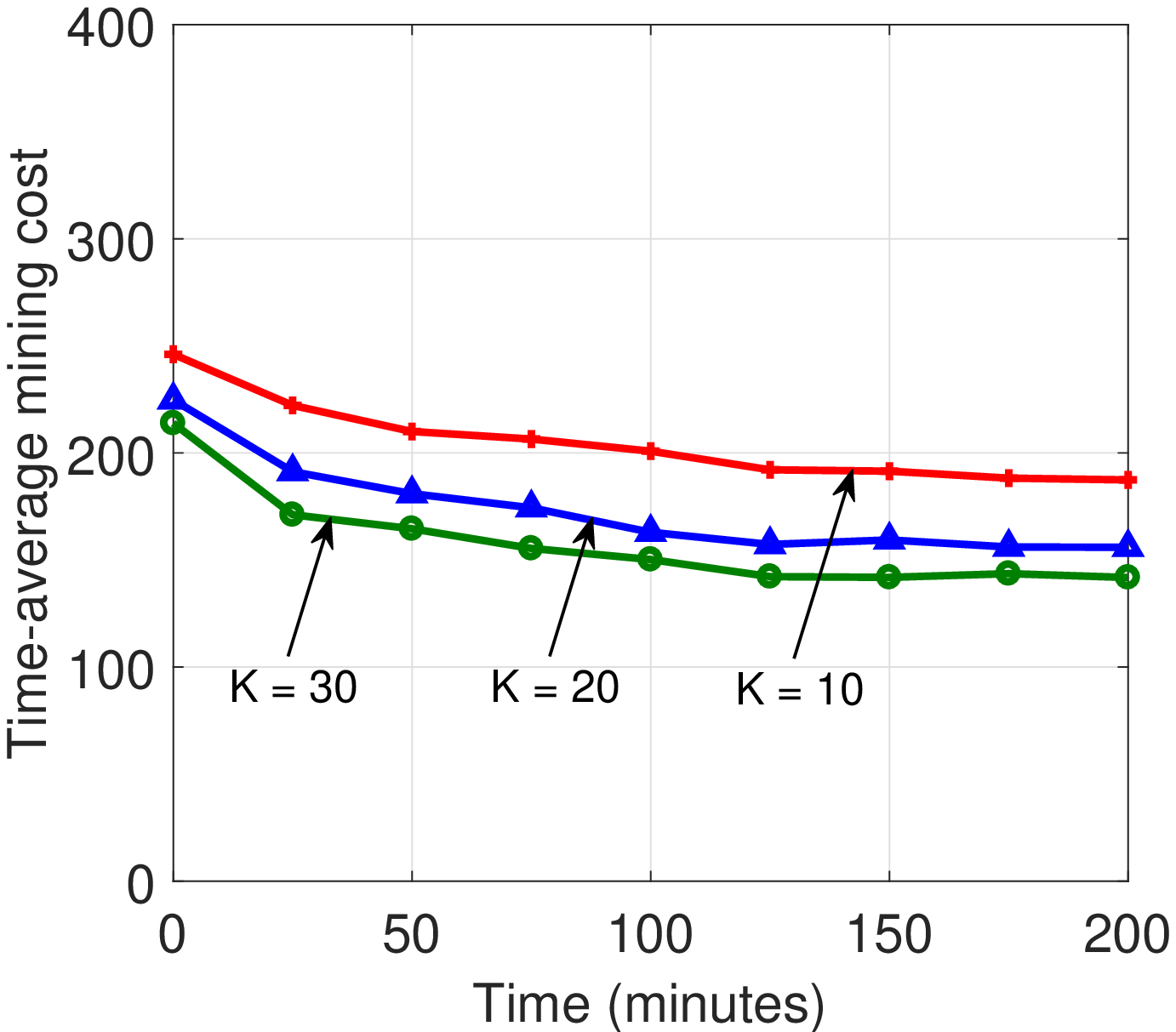}}
		\caption{Time-average mining cost for different $K$.} 
		\label{fig_mine_cost_K}
	\end{minipage}%
	\hspace{0.009\textwidth}
	\begin{minipage}[t]{0.47\linewidth}
		{\includegraphics[height=3.6cm]{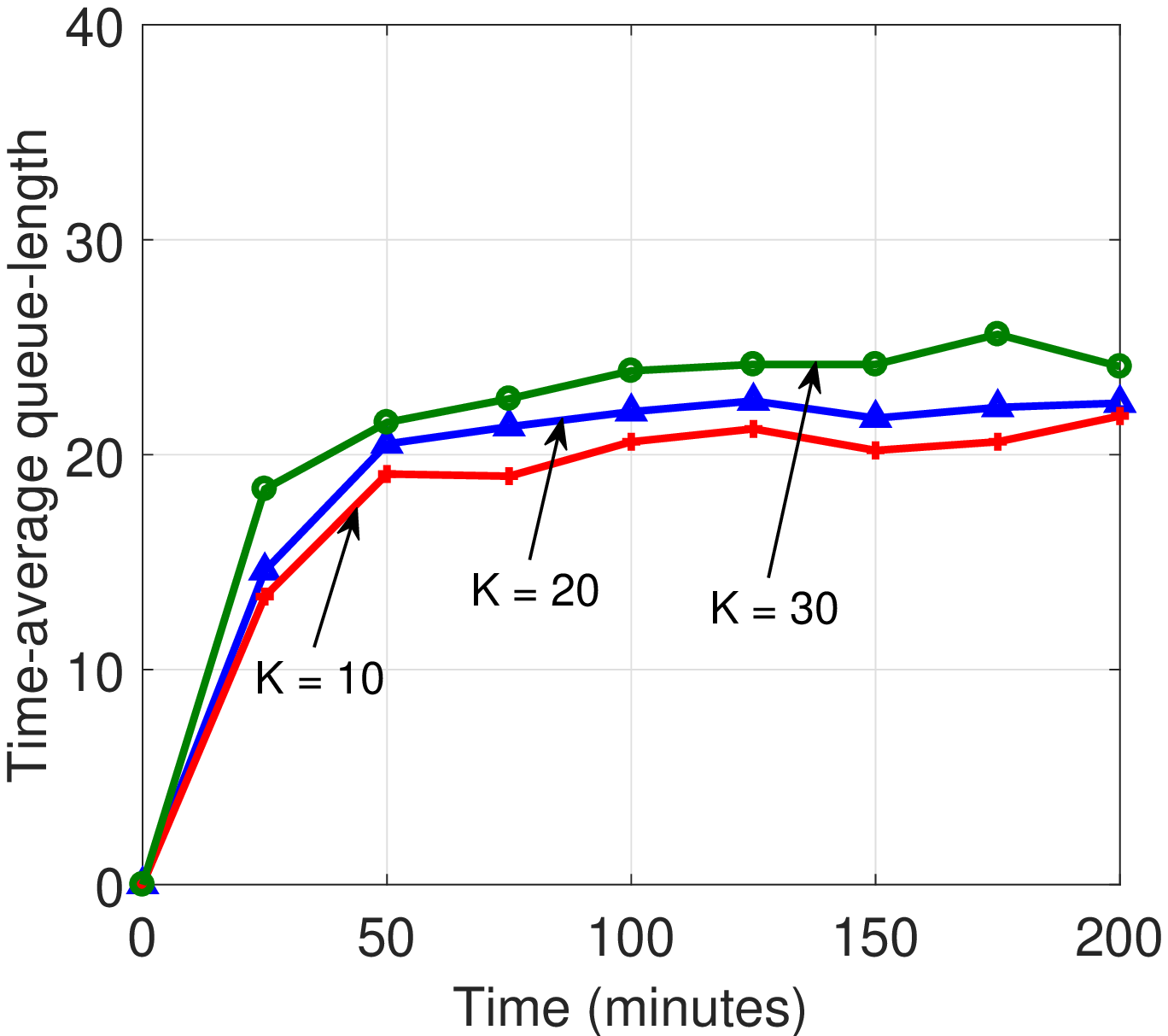}}
		\caption{Time-average queue-length for different $K$.} 
		\label{fig_mine_queue_K}
	\end{minipage}
	\vspace{-.15in}
\end{figure}

{\em 3) Impact of the Tradeoff Parameter $K$}:
Figs.~\ref{fig_mine_cost_K} and~\ref{fig_mine_queue_K} show the tradeoff between time-average mining cost and time-average queue-length under different values of tradeoff parameter $K$. 
We can observe that a larger $K$ leads to a better mining cost performance, but at the cost of incurring a larger queueing delay and vice versa.
Thus, there is a tradeoff between the mining cost and the queueing delay.
By tuning the value of $K$, we can control the tradeoff between mining cost and delay in blockchain networks.

\section{Conclusion} \label{sec:conclusion}
In this paper, we proposed a queueing-based analytical model to study the mining resources allocation problem in PoW-based blockchain networks. 
In our queueing model, we did not assume the knowledge of arrivals probability distribution and are agnostic to any priority mechanism. 
We leveraged Lyapunov optimization techniques and proposed a dynamic mining resources allocation algorithm with a parameter $K$, which can be used to tune the trade-off between mining energy cost and queueing delay. 
We showed that the proposed algorithm is within an $O({1 \mathord{\left/{\vphantom {1 K }} \right.\kern-\nulldelimiterspace} K })$ distance to the optimum, with a worst case delay scaling as $O(K)$. 
Our simulation results also demonstrated the effectiveness of our proposed algorithm in reducing the mining cost for blockchain networks.


\begin{thebibliography}{10}
\providecommand{\url}[1]{#1}
\csname url@samestyle\endcsname
\providecommand{\newblock}{\relax}
\providecommand{\bibinfo}[2]{#2}
\providecommand{\BIBentrySTDinterwordspacing}{\spaceskip=0pt\relax}
\providecommand{\BIBentryALTinterwordstretchfactor}{4}
\providecommand{\BIBentryALTinterwordspacing}{\spaceskip=\fontdimen2\font plus
\BIBentryALTinterwordstretchfactor\fontdimen3\font minus
  \fontdimen4\font\relax}
\providecommand{\BIBforeignlanguage}[2]{{%
\expandafter\ifx\csname l@#1\endcsname\relax
\typeout{** WARNING: IEEEtran.bst: No hyphenation pattern has been}%
\typeout{** loaded for the language `#1'. Using the pattern for}%
\typeout{** the default language instead.}%
\else
\language=\csname l@#1\endcsname
\fi
#2}}
\providecommand{\BIBdecl}{\relax}
\BIBdecl

\bibitem{nakamoto2008bitcoin}
S.~Nakamoto, ``Bitcoin: A peer-to-peer electronic cash system,'' 2008.

\bibitem{kamath2018food}
R.~Kamath, ``Food traceability on blockchain: Walmart’s pork and mango pilots
  with ibm,'' \emph{The Journal of the British Blockchain Association}, vol.~1,
  no.~1, p. 3712, 2018.

\bibitem{o2014bitcoin}
K.~J. O'Dwyer and D.~Malone, ``Bitcoin mining and its energy footprint,'' 2014.

\bibitem{bentov2016snow}
I.~Bentov, R.~Pass, and E.~Shi, ``Snow white: Provably secure proofs of
  stake,'' \emph{IACR Cryptology ePrint Archive}, vol. 2016, p. 919, 2016.

\bibitem{zhang2017rem}
F.~Zhang, I.~Eyal, R.~Escriva, A.~Juels, and R.~Van~Renesse, ``{REM}:
  Resource-efficient mining for blockchains,'' \emph{IACR Cryptology ePrint
  Archive}, vol. 2017, p. 179, 2017.

\bibitem{ahmed2018don}
M.~Ahmed and K.~Kostiainen, ``Don't mine, wait in line: Fair and efficient
  blockchain consensus with robust round robin,'' \emph{arXiv preprint
  arXiv:1804.07391}, 2018.

\bibitem{gervais2016security}
A.~Gervais, G.~O. Karame, K.~W{\"u}st, V.~Glykantzis, H.~Ritzdorf, and
  S.~Capkun, ``On the security and performance of proof of work blockchains,''
  in \emph{Proceedings of the 2016 ACM SIGSAC Conference on Computer and
  Communications Security}.\hskip 1em plus 0.5em minus 0.4em\relax ACM, 2016,
  pp. 3--16.

\bibitem{kasahara2016effect}
S.~Kasahara and J.~Kawahara, ``Effect of bitcoin fee on
  transaction-confirmation process,'' \emph{arXiv preprint arXiv:1604.00103},
  2016.

\bibitem{ricci2019learning}
S.~Ricci, E.~Ferreira, D.~S. Menasche, A.~Ziviani, J.~E. Souza, and A.~B.
  Vieira, ``Learning blockchain delays: A queueing theory approach,'' \emph{ACM
  SIGMETRICS Performance Evaluation Review}, vol.~46, no.~3, pp. 122--125,
  2019.

\bibitem{koops2018predicting}
D.~Koops, ``Predicting the confirmation time of bitcoin transactions,''
  \emph{arXiv preprint arXiv:1809.10596}, 2018.

\bibitem{jiao2018social}
Y.~Jiao, P.~Wang, D.~Niyato, and Z.~Xiong, ``Social welfare maximization
  auction in edge computing resource allocation for mobile blockchain,'' in
  \emph{ICC}.\hskip 1em plus 0.5em minus 0.4em\relax IEEE, 2018, pp. 1--6.

\bibitem{xiong2018optimal}
Z.~Xiong, S.~Feng, D.~Niyato, P.~Wang, and Z.~Han, ``Optimal pricing-based edge
  computing resource management in mobile blockchain,'' in \emph{ICC}.\hskip
  1em plus 0.5em minus 0.4em\relax IEEE, 2018, pp. 1--6.

\bibitem{li2019decentralized}
J.~Li, Z.~Zhou, J.~Wu, J.~Li, S.~Mumtaz, X.~Lin, H.~Gacanin, and S.~Alotaibi,
  ``Decentralized on-demand energy supply for blockchain in internet of things:
  A microgrids approach,'' \emph{IEEE Transactions on Computational Social
  Systems}, 2019.

\bibitem{dimitri2017bitcoin}
N.~Dimitri, ``Bitcoin mining as a contest,'' \emph{Ledger}, vol.~2, pp. 31--37,
  2017.

\bibitem{neely2010stochastic}
M.~J. Neely, ``Stochastic network optimization with application to
  communication and queueing systems,'' \emph{Synthesis Lectures on
  Communication Networks}, vol.~3, no.~1, pp. 1--211, 2010.

\bibitem{Tassiulas92:BackPressure}
L.~Tassiulas and A.~Ephremides, ``Stability properties of constrained queuing
  systems and scheduling policies for maximum throughput in multihop radio
  networks,'' vol.~37, no.~12, pp. 1936--1948, Dec. 1992.

\end{thebibliography}

\end{document}